\documentclass[twocolumn]{article}
\usepackage{usenix2019}

\usepackage{authblk}
\author[1]{Saar Tochner}
\author[2]{Stefan Schmid}
\affil[1]{The Hebrew University of Jerusalem}
\affil[2]{Faculty of Computer Science, University of Vienna}

\usepackage{hyperref}
\usepackage{natbib}
\usepackage{graphicx}
\usepackage{xcolor}
\usepackage{amsthm}
\usepackage{amsfonts}
\usepackage{graphicx}

\setcitestyle{numbers}

\newtheorem{lemma}{Lemma}

\newtheorem{definition}{Definition}

\newtheorem{example}{Example}
\newcommand{\boldheader}[1]{\vskip 5pt \noindent{\bf #1}}

\begin{document}

\title{Routing with Selfish Knowledge:\\
Making trampoline nodes scale \& secure}

\title{On Confidential Routing\\ in Large-Scale Offchain Networks}

\title{On Search Friction of Route Discovery in Offchain Networks}

\maketitle


\begin{abstract}
Offchain networks provide a promising solution to overcome
the scalability challenges of cryptocurrencies.
However, design tradeoffs of offchain networks are still
not well-understood today. 
In particular, offchain networks typically 
rely on fees-based incentives
and hence require mechanisms for the 
efficient discovery of ``good routes'':
 routes with low fees (cost efficiency)
 and a high success rate of the transaction routing
 (effectiveness). 
 Furthermore the route discovery should be confidential (privacy),
and e.g., not reveal information about who transacts with whom
or about the transaction value. 
This paper provides  an analysis of the 
``search friction'' of route discovery,
i.e., the costs and tradeoffs of route discovery in 
large-scale offchain networks in which nodes
behave strategically. 
As a case study, we consider the Lighning network
and the route discovery service
provided by the trampoline nodes,
evaluating the tradeoff in different scenarios
also empirically. 
Finally, we initiate the discussion of alternative charging schemes
for offchain networks.
\end{abstract}




\section{Introduction}

Despite the high popularity of cryptocurrencies, 
it remains a challenge to make fast payments at scale.
This is mainly due to the inefficiency of the underlying
consensus protocol: it can take several minutes
until a transaction went through a full consensus 
and can be confirmed.
A promising solution are emerging payment channel networks
such as the Lightning network,
which allow to perform transactions off-chain and
in a peer-to-peer fashion: 
without requiring consensus on the blockchain.
In a nutshell, a payment channel is a cryptocurrency transaction 
which escrows or dedicates money
on the blockchain for exchange with a given user and duration. 
Users can also interact if they do not share a direct payment channel:
they can route transactions through \emph{intermediaries}. 

However, the design of 
secure and scalable offchain networks is 
challenging and still not well-understood.
In particular, these networks must not only
be scalable but also account for strategic (i.e., selfish)
user behavior; it must further be ensured that
these networks do not introduce new 
security issues.
A common approach to
incentivize network nodes (the intermediaries) to contribute
to the transaction routing
is to use a fee-based mechanism: 
intermediaries can charge nodes which route through them a
nominal fee. 
This is also the approach taken in the 
Lightning network which serves us as a case study in
this paper.

This raises the question of how nodes
can \emph{discover} routes through intermediaries.
One aspect here is cost efficiency:
since different routes come at different
fees, nodes require scalable mechanisms to
find ``short'' (i.e., low-cost) routes.
However, routes do not only have to be
cheap but also provide sufficient liquidity
to route the transaction: 
the route discovery mechanism should ensure a
high success rate of the transaction routing;
this property is known as effectiveness in
the literature. 
Effectiveness is not only a performance concern:
a lengthy discovery process may also jeopardize
privacy, potentially leaking information about
who aims to transact (i.e., find a route) with whom.
Last but not least, the route discovery should be 
incentive-compatible, e.g., account
that nodes are only willing to distribute routes
from which they can benefit (e.g., which go through
themselves).

Providing an effective and scalable route discovery
is particularly challenging as large-scale off-chain networks are 
expected to be highly dynamic, 
e.g., due to the frequent changes of channels and fees.
This renders solutions
requiring wallets to keep up-to-date state information
about the networks infeasible. 
An interesting recent solution to reduce the burden
on wallets, is the deployment of route discovery servers, such 
as the trampoline nodes in Lightning: 
these servers maintain routes so that a wallet just
needs to know how to reach the route server nodes
in its neighborhood and can then request
the desired route. 

This paper provides an analysis of the efficiency-privacy
tradeoff of off-chain route discovery, considering
the Lightning network as a case study. In particular, 
we investigate to which extent route discovery
can be efficient and effective, incentive compatible and 
confidential.
Here, confidentiality is about more than just the actual
data that is communicated in the discovery process, e.g.,
the source, the destination, or the transaction size;
it is also about the possible metadata that is communicated
implicitly, e.g., about the rate or time at which transactions
occur. In fact, existing cryptographic techniques 
such as \cite{hu2014private} can be used
to provide data confidentiality, however,
as we will show in this paper, nodes may
still leak information about the frequency of transactions,
i.e., about their \emph{transaction rate}, to other nodes which are
not on the transaction route.

We quantify the ``search friction'', i.e., the cost
of the route discovery process, both analytically, deriving
cost lower bounds, as well as empirically, considering real
payment channel networks. 
Our results motivate research into alternative economic
models to provide routing incentives which come at lower
search friction costs, which we also start to discuss in this paper. 

The remainder of this paper is organized as follows,
see also Figure~\ref{figure::venn_organize}.
We introduce a model for route discovery 
and provide a formal analysis in Section~\ref{sec::model},
and then report on our empirical results in 
Section~\ref{sec::experiments}. 
After discussing related work in 
Section~\ref{sec::related_work}, we conclude our contribution in
Section~\ref{sec::conculsions}. 

\begin{figure}
	\centering
	\includegraphics[scale=0.5]{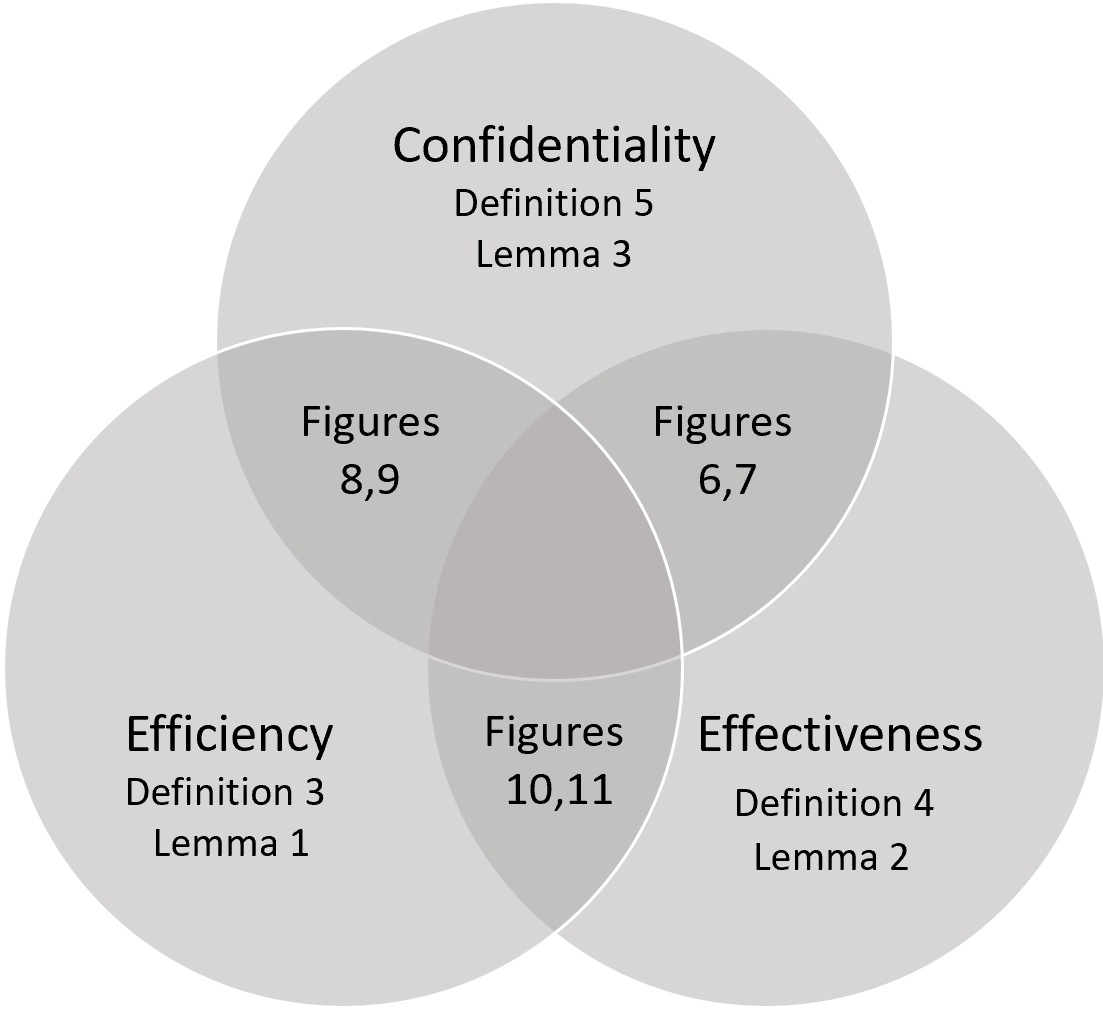}
	\caption{Paper organization}
	\label{figure::venn_organize}
\end{figure}

\section{Tradeoff Analysis of Route Discovery} \label{sec::model}

In order to analyze the search friction and tradeoffs of route
discovery, we consider the following simplified model.
We model the offchain network 
as a graph of \emph{channels}:
two nodes can create (and may later delete) a payment channel
between each other, to perform offchain transactions.
This channel may also be used by other nodes for
routing their transactions indirectly, i.e., using
multi-hop routing.
In this paper, we are primarily interested in two properties of a
channel: the fee other nodes are charged to use this channel,
and the channel capacity (resp. liquidity), which determines the size
of the transactions it can support. This is a simplification
of actual systems such as Lightning, but captures
their essence.\footnote{E.g., nodes can set min/max
values for channels, create non-public channels which
cannot be used for multihop routing, etc.}

\begin{definition}[Offchain Network]
    The offchain network is 
    a weighted, directed graph $G(V,E)$, where $V$ are the nodes in the network 
    and $E$ are the channels. A channel $e \in E$ is characterized 
    by a weight  $w_e \in \mathbb{R_{+}}$ 
    and by a capacity $w_{cap}$.
\end{definition}

The weight $w_e$ is kept general here but typically is a function of 
the fee that the source node pays for an intermediate channel, 
the channels' age (older channels may be assumed to be more reliable), 
or previous knowledge (e.g., channels which failed in the past), among other. 
    
The network is dynamic, i.e., channels may be added and removed over time,
fees updated, etc., and hence nodes require mechanisms to
discover the topology and learn about updates to be able to route
their transactions through the offchain network, either explicitly or implicitly.  For example,
Lightning includes gossiping and active probing mechanisms
to allow nodes to learn about routing fees. It is
however more challenging to learn about the capacity $w_{cap}$;
such information is typically not distributed for privacy reasons
and hence, finding a path with sufficient capacity may require
trial and error \cite{icissp20}.

More formally, a route on a topology $G(V,E)$ from $s \in V$ to $t \in V$ is a list of channels $e_1. \cdots, e_n \in E$ such that the source node of $e_1$ is $s$, and the target node of $e_n$ is $t$ and for every $i$, the target of $e_i$ is the source of $e_{i+1}$. 
A transaction in the network is a payment from a source to a target along a 
``valid route'': A route can serve a transaction of size $l$ if every channel $e$ along the route has enough capacity, i.e. $e_{cap} \ge l$. 
We assume that the weight of a route is simply the sum of the weights of 
its channels.

Inspired by existing offchain networks such as Lightning,
we distinguish between two types of nodes:
wallets and ``regular nodes'', henceforth called trampoline nodes (TNs). 
Wallets are simple nodes and do not have the resources
to store (and maintain!) information about the entire network.
Rather, they need to rely on the trampoline nodes which know
the network and which may inform wallets about routes upon
request. We are interested in exactly this discovery process,
where wallets interact with one or multiple trampoline nodes
to find routes for their later transactions.

This route discovery process however introduces
the following challenges:
\begin{itemize}
      \item \emph{Strategic behavior and efficiency:} Trampoline nodes may
        act selfishly and may only have an incentive
        to share routes which include themselves,
        such that they can charge the fee.
        As a consequence, a wallet may not learn about
        the most efficient (i.e., lowest cost) route.
        
      \item \emph{Effectiveness:} Also related to the above,
      wallets may have to invest more resources into the discovery
      of efficient routes, exploring additional alternative trampoline
      nodes. The effectiveness of this route discovery process
      is further affected by the fact that not all the discovered
      routes may provide sufficient liquidity (i.e., capacity)
      for a large transaction which needs to be routed.
        
     \item \emph{Privacy:} Through the repeated interactions
     with multiple trampoline nodes, querying for
     specific routes, a wallet may reveal
     confidential information about its transactions.
\end{itemize}

We are interested in the following family of route
discovery algorithms:
\begin{definition}[Routing Discovery Algorithm (RDA)]
        A $q$-route discovery algorithm (RDA) 
        is an algorithm that given a pair $s, t \in V$,
        performs at most $q$ queries, issued to $q$ trampoline nodes, 
        and either returns a valid route or decides that this
				is not possible.
\end{definition}

We measure the quality, i.e., the efficiency, of a route found by the RDA, 
by comparing it to the \emph{optimal route} with respect to the 
weight function on the topology. 

\begin{definition} [Efficiency]
    The efficiency of a route $R$ from $s$ to $t$ is
		defined by the stretch, i.e., $\frac{w(R_{src, dst})}{w(O_{src, dst})}$, 
		where $w(\cdot)$ is the weight of the route, and $O_{src, dst}$ is the route with the minimal weight between  $s$ to $t$.
		The weight of a route is the sum of its link weights.
\end{definition}

As discussed above, some routes in off-chain networks may be temporarily 
unavailable (e.g., due to offline nodes or lack of liquidity) and thus invalidate the result of the RDA.
Another important metric to evaluate RDA hence concerns the number
of queries it needs to issue until a valid route
is discovered. 
For example, in the Lightning network, an available route is searched
as part of the route initialization procedure. This process locks 
the channels along the route (a designated amount) and 
the channel commits 
to participate in the transaction.
\begin{definition} [Effectiveness]
   The effectiveness of an RDA is 
	the number of queries which have to be issued
	to successfully execute a given transaction.
\end{definition}

Furthermore, as transactions are privacy critical,
the RDA should 
not leak any confidential information.
Naturally, a first concern regards the information provided 
by the query directly,
including e.g., source and destination nodes,
potentially the transaction size, the resulting routes, etc.
As discussed above, today we understand fairly well
how to protect such information, e.g.,
using homomorphic encryption schemes
and private information retrieval \cite{morais2018zero,homenc,de2010practical,hu2014private}.
However, there is another concern, related to
the meta-data revealed from the query, 
e.g., 
        the timestamp or even the existence of the 
				route discovery query itself. 
While there also exist solutions to metadata private messaging 
systems, e.g., \cite{lazar2018karaoke,lazar2016alpenhorn,tyagi2017stadium}), we
will show in the following that there is an inherent limitation
what can be achieved in terms of an efficient and confidential 
route discovery.
To this end, we introduce the notion of 
\emph{leak rate}:
to what extent can a node learn about the number
of transactions in the network in a given time unit?
That is, the leak rate is defined
    as the number of transactions
    in a single time unit that a node can learn about
    for a given set of transactions $T$ 
    under a given route discovery algorithm $A$.
\begin{definition}[Leak Rate]
    An RDA $A$ leaks at rate $k$ if 
    in order to route a transaction,
		$k$ times more nodes will learn about the existence
		of this transaction compared to a scenario
		where the transaction is simply routed along the
		shortest path (e.g., using source routing). 
\end{definition}

To clarify and motivate this notion,
we give an example in 
Figure \ref{figure::reate_leaks_illustration}. 
In this simple network, a node 
learns about a transaction
it should in principle have no idea about. 
\begin{figure}
	\centering
	\includegraphics[scale=0.5]{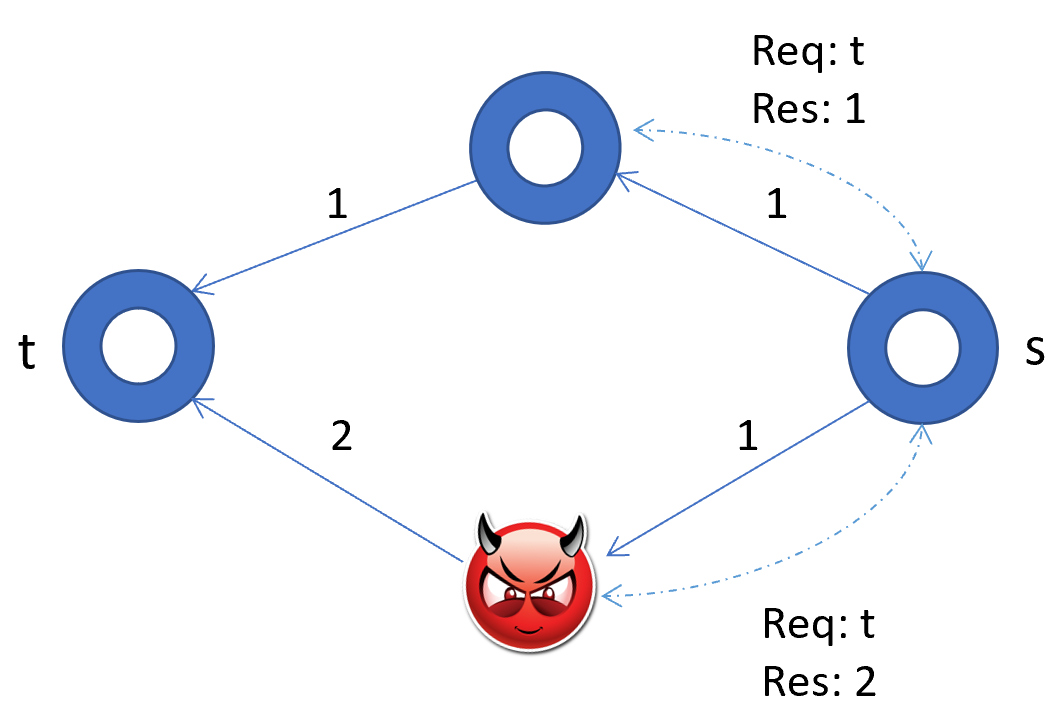}
	\caption{An attacker can learn about 
	a transaction although 
	the request itself does not hold any information.}
	\label{figure::reate_leaks_illustration}
\end{figure}

\boldheader{Efficiency.}
With these concepts in mind, we now 
first analyze the efficiency achievable
by route discovery algorithms.
The following lemma shows an inherent
cost of the route discovery process
in the off-chain model
\begin{lemma} \label{lemma:efficiency}
    For every $q$-RDA and every $M \in \mathbb{R}$ 
		there exists a topology in which an RDA will return a route with weight 
		$M$ times higher than the optimal route, or it will not return a route at all.
\end{lemma}
\begin{proof}
    Consider 
		the topology in Figure~\ref{figure::non_optimal_illustration}.
    Given a $q$-RDA $A$, we build a topology in which $A$ will return a route 
		with weight larger than $M$ although there exists a route with weight $1$.
    In our topology, the source node, $s$, is in the center, and is connected 
		to $q+1$ TNs with channel weight $0$. Each TN is connected to one unique node with a channel of weight $1$, and all these nodes form a 
		clique (i.e., are connected to one another) by channels 
		of weight $M$. We will choose the target from one of these nodes.
    The RDA $A$ queries $q$ TNs in an order that is independent on target node (because $s$ does not know the topology). But there are $q+1$ possibilities to the target, therefore there exists a node in the outer circle, $t$, that $A$ does not query its direct TN neighbor. Choose this $t$ to be our target. As the TNs are selfish, therefore they will tell $A$ only about routes that go through themselves, and all of them are not directly connected to $t$; thus all the weights that $A$ sees are at least of weight $M+1$.
    Finally, $A$ will return either a route with weight $M+1$ or no route at all, although the actual shortest route is of weight $1$.
\end{proof}

\begin{figure}
	\centering
	\includegraphics[scale=0.5]{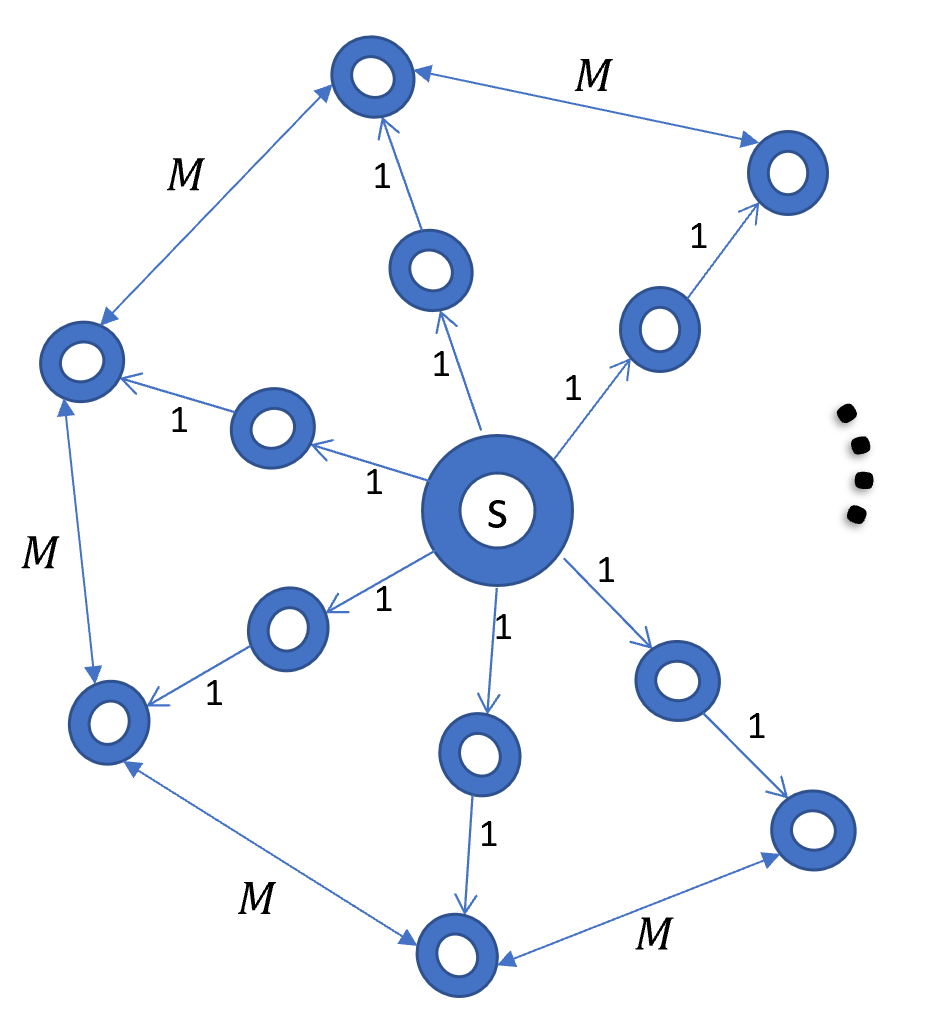}
	\caption{Example with high cost: If there are $q+1$ direct neighbors, 
	then any $q$-RDA struggles to find an efficient route.}
	\label{figure::non_optimal_illustration}
\end{figure}

In general, the efficiency will depend on the specific
topology. To give an example, consider the complete
network.

\begin{example}
In a clique where all the weights are equal, 
the efficiency, in terms of the stretch, is upper bounded by 2.
\end{example}

To see this, assume $r$ is the weight of each channel. 
The optimal route is the direct channel, which weighs $r$. 
On the other hand, the route to each TN and the route from the TN to each target are also $r$ (this is clique, so there exists a direct channel), 
which is $2r$ in total.

\boldheader{Effectiveness.}
We next consider the effectiveness
a route discovery algorithm can achieve. 
Also here, we first derive a negative result
for the general scenario.
\begin{lemma} \label{lemma:effectiveness}
    For every $q$-RDA
		and every $M \in \mathbb{R}$, there exists a topology in which the first $M$ routes from the algorithm will be unsuccessful.
\end{lemma}
\begin{proof}
    As in Lemma~\ref{lemma:efficiency}, we will build a topology in which the effectiveness of the RDA $A$ will be as needed.
    The topology is composed from $M+2$ nodes: $M$ nodes form a clique, one of the nodes in the clique is a TN, to which the source is connected, and the target is connected to other node from the clique.
    There are $\sum_{k=0}^{M-2} {M-2 \choose k}k!$ unique routes in the clique. $A$ is not aware to the availability of the channels in the topology, therefore it offers the routes in an order which is independent to it. Now, define the channels to be unavailable in any of the first $M$ routes.
    Therefore the first $M$ routes that $A$ will offer will be unavailable.
\end{proof}

Again, more specific networks can provide
better guarantees.

\begin{example}
In a scale-free network with $n$ nodes, 
where all the channels are bi-directional, and where $p$ 
is the probability that a channel has sufficient capacity, 
an $q$-RDA that queries the highest degrees nodes succeeds with 
a probability of at least $1-(1 - p^{ \frac{2 \cdot log(n)}{\log\log(n)}})^{n \cdot (1 - 2 ^ {-q})}$.
For example, if a channel accepts a route independently w.p. $p=0.2$, 
for a network of size $n=4000$ (e.g., Lightning), when the RDA queries only $q=5$ TNs, then the probability to get at least one effective route is $\ge 0.999$.
\end{example}

To see this, note that the diameter of this network is $\frac{\log(n)}{\log\log(n)}$ (following \cite{bollobas2004diameter}), therefore the length of the path from every source to every TN, and then from the TN to the target is bounded by $2 \cdot \frac{\log(n)}{\log\log(n)}$. Therefore, the probability for each TN's suggestion to route the transaction is at least $p^{2 \cdot \frac{\log(n)}{\log\log(n)}}$. Moreover, the number of paths from TN to a target is at least its degree, and the total degrees of the $q$ nodes with the largest degree is  $n \cdot (\frac{1}{2} + \frac{1}{4} + \cdots + \frac{1}{2^{q}}) = n \cdot (1 - 2^{-q})$.

\boldheader{Confidentially.}
To which extent can we avoid rate leakage
of the route discovery algorithm?
The following lemma shows that 
if nodes behave strategically, 
we cannot upper bound the leaking rate of a
route discovery algorithm.

\begin{lemma}
    For every $q$-RDA, and every $M \in \mathbb{R}$ there exists a topology in which the algorithm leaks at rate $\min\{M,q\}$.
\end{lemma}
\begin{proof}
    The proof is 
		similar to the proof of Lemma~\ref{lemma:efficiency},
		however, rather than worrying about 
		the weight of the discovered route, 
		we only worry about the existence of a valid route. 
		Therefore, we will define our topology with
		a set of disconnected nodes,
		among the nodes in the outer circle (instead of clique).
    More specifically, our topology is simply a star topology, 
		where the source is in the center, its neighbors are $M$ TNs, 
		and their neighbors are the possible targets (and they are all sinks, i.e. without outgoing channels).
    Clearly, if the RDA $A$ will query a TN which is not the direct neighbor of the target, then it will return nothing (because there is no route). Let us choose our target to be the node that is connected to the $M$'th TN that $A$ queries. If $M<q$ then the $A$ will query $M$ TNs until it will find a valid route; 
		otherwise it will stop after querying $q$ nodes unsuccessfully.
    In either way, $A$ queries $\min\{M,q\}$ TNs.
\end{proof}

\begin{example}
In scale free networks, the number of queries that we need to perform in order to find a ``good" route is small.
\end{example}

For example, in terms of betweenness, in a scale free network with $n$ nodes, a node with degree 
$n \cdot 2^{-k}$ participates in $2^{-k}$ of the shortest routes \cite{kitsak2007betweenness}. Therefore if the highest degree nodes are TN, and we ask the top $k$, then the TN will be on the optimal route to the target w.p. $1 - 2^{-k}$. Therefore the number of queries that we need to do in order to find an optimal route is exponentially small compared to the number of nodes in the topology.


\section{Empirical Evaluation} \label{sec::experiments}

In order to complement our theoretical results and in order
to study the efficiency, effectiveness and confidentiality tradeoff
in real networks, we consider the Lightning network as a case study
for our experiments in the following. 
 
\subsection{Methodology}

Following the Lightning network RFC regarding trampoline nodes\footnote{
see \url{https://github.com/lightningnetwork/lightning-rfc/pull/654}}, we assume that a wallet node stores its local knowledge on close trampoline nodes (TN) in the network, and can search for them  within a close neighborhood. The wallet then queries some of the TNs for a route to the desired target, 
and will finally use the route with the lowest weight.

\noindent \textbf{Collected data.}
We collected data about the Lightning network
using a live Lightning node (lnd) which 
is connected to the mainnet through bitcoind.
To extract the network structure (currently, the whole topology is stored 
by all nodes) we use the command $lncli~ describegraph$.
The data used for this paper 
was retrieved on March 24th 2020,
and provides information 
about the public channels. 

\noindent \textbf{Topologies.}
We consider two topologies in our experiments: 
The first is the Lightning network snapshot, and 
the second is a synthetic 
sparse topology that was built using the following algorithm: 
create 1000 nodes in a circle, add one outgoing channel from each node to the node next to it, and another channel to other random node (total of 2 outgoing channels). The fee is the same for all the channels and equals to 1. 
Figures~\ref{figure::sparse_neighborhood}, \ref{figure::lightning_neighborhood} show the number of nodes as a function of the neighborhood size in the two topologies.

\begin{figure}
	\centering
	\includegraphics[scale=0.5]{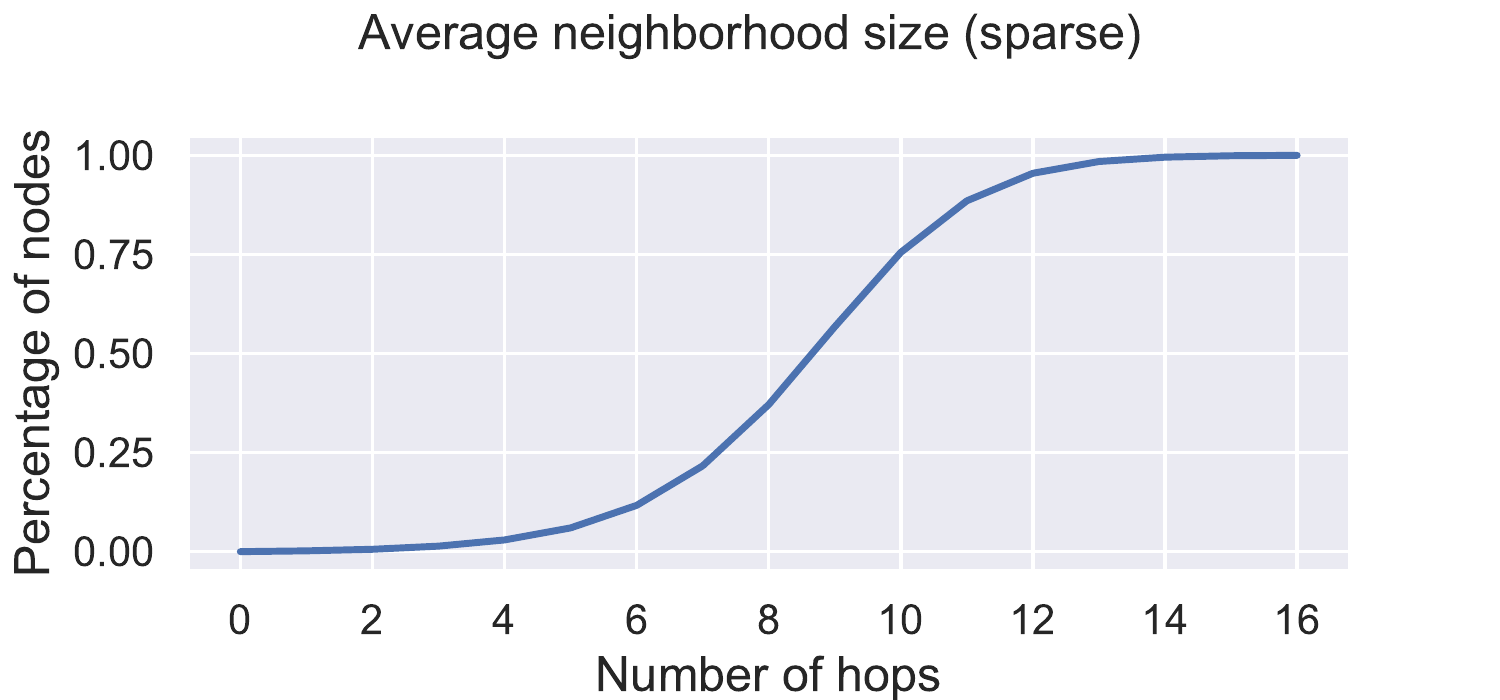}
	\caption{The percentage of nodes as a function of the neighborhood size in sparse topology.}
	\label{figure::sparse_neighborhood}
\end{figure}
\begin{figure}
	\centering
	\includegraphics[scale=0.5]{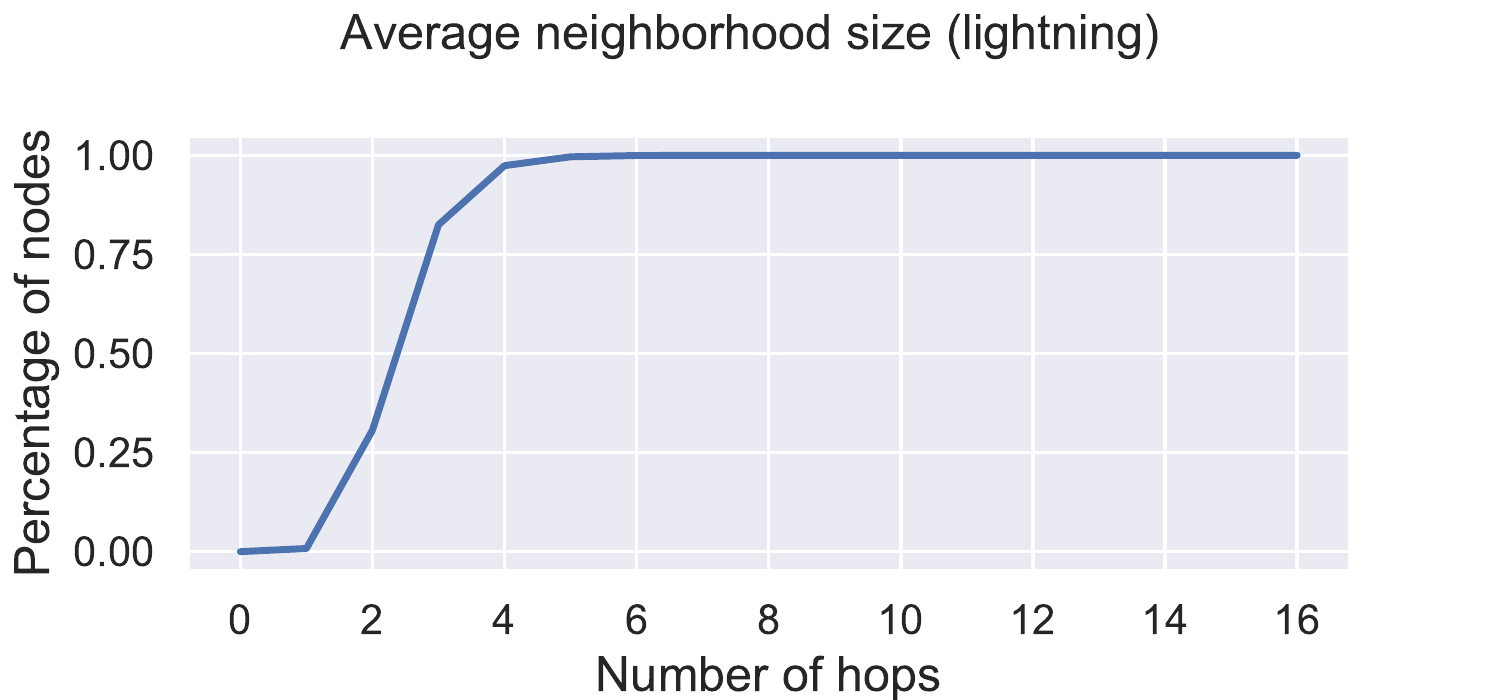}
	\caption{The percentage of nodes as a function of the neighborhood size in the Lightning network's topology.}
	\label{figure::lightning_neighborhood}
\end{figure}

\noindent \textbf{Graph weights.}
The weights of the edges in the Lightning network's graph are
determined by the transaction size. Each channel determines base and proportional fee, and the final fee is $base\_fee + transaction\_size \times proportional\_fee$. Moreover, the fees are calculated backward from the target to the source, because the nodes should pay the fee for transferring the fees to the later intermediate nodes. 
We decided to neglect this backward computation due to the fact that practically (following our previous works) it does not change the routes. In addition, we determined the transactions' size to be $10^{6}$ milisatoshis.
We finally determine the weight of each channel to be $base\_fee + 10^6 \times proportional\_fee$.

\noindent \textbf{Transactions distribution.}
In the following, we assumed that transactions follow two possible distributions: (i) one where all pairs of nodes in the network attempt to create a transaction and (ii) one where there are nodes with higher probability to execute a transaction (higher ``activity level"). In the latter case, we determine a power-law distribution and grant it uniformly to the nodes. In particular, we uniformly partition the nodes to groups of size 100, and give the $i$'th group an activity level of $2^{-i}$.
Note that we need to model the transactions because transactions in the Lightning network are private by design. It is hard to infer the real distribution of the transactions since (i) information about transactions is hidden in the private state of channels and since (ii) routes are obscured by onion encryption. 

\noindent \textbf{Implementation details.}
We used the Floyd–Warshall all pairs shortest path
algorithm to compute the optimal routes. 
Moreover, in order to keep shortest paths by limiting the number of neighborhood sizes (for example to search for the weight to the trampoline nodes in a neighborhood with a certain size), we used the ``min-plus matrix multiplication" (or ``distance product") algorithm, and stored the weights matrices for each number of hops. Finally, to find all the paths between specific source to target, we used the python module $networkx$.

\subsection{Tradeoff Evaluation}

\boldheader{Efficiency-Confidentiality Tradeoff}
We first evaluate the efficiency of the routes,
depending on the neighborhoods in which trampoline nodes are searched. 
We already know from Lemma~\ref{lemma:efficiency} that the efficiency 
can be low in the worst case, and we are now interested the
efficiency in our specific examined topologies. 
The efficiency-confidentiality tradeoff 
can help a wallet to decide on the neighborhood size that 
it should query in order to find an efficient route.

We first consider our 
synthetic sparse topology: 
Figure~\ref{figure::sparse_uniform_2d} 
shows that the efficiency 
in small neighborhoods is better only because we cannot find many routes
in this scenario. In the sparse topology this makes sense: 
there are less edges, therefore if there is a TN in a close neighborhood, then the optimal route goes through it with higher probability. 
For the Lightning network topology, in 
Figure~\ref{figure::lightning_uniform_2d}. 
we can see that nodes find more TNs in their close neighborhoods, 
but the weights are far from  optimal. 
One possible reason is that the topology is scale-free, 
therefore close neighborhoods are crowded, and 
the routes and the TN that the node finds are not on the optimal route.

\begin{figure}
	\centering
	\includegraphics[scale=0.5]{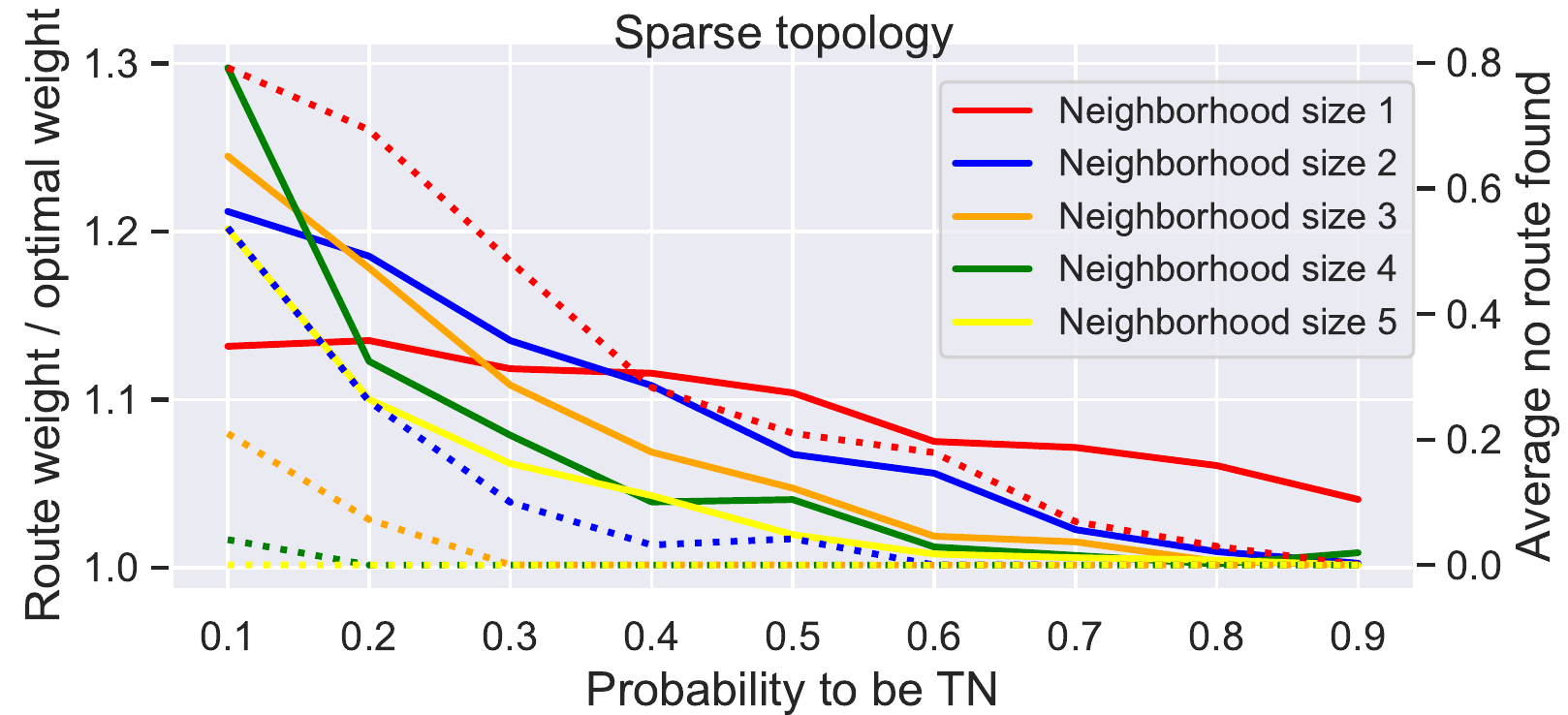}
	\caption{The probability to not find a route compared to the found route.}
	\label{figure::sparse_uniform_2d}
\end{figure}

\begin{figure}
	\centering
	\includegraphics[scale=0.5]{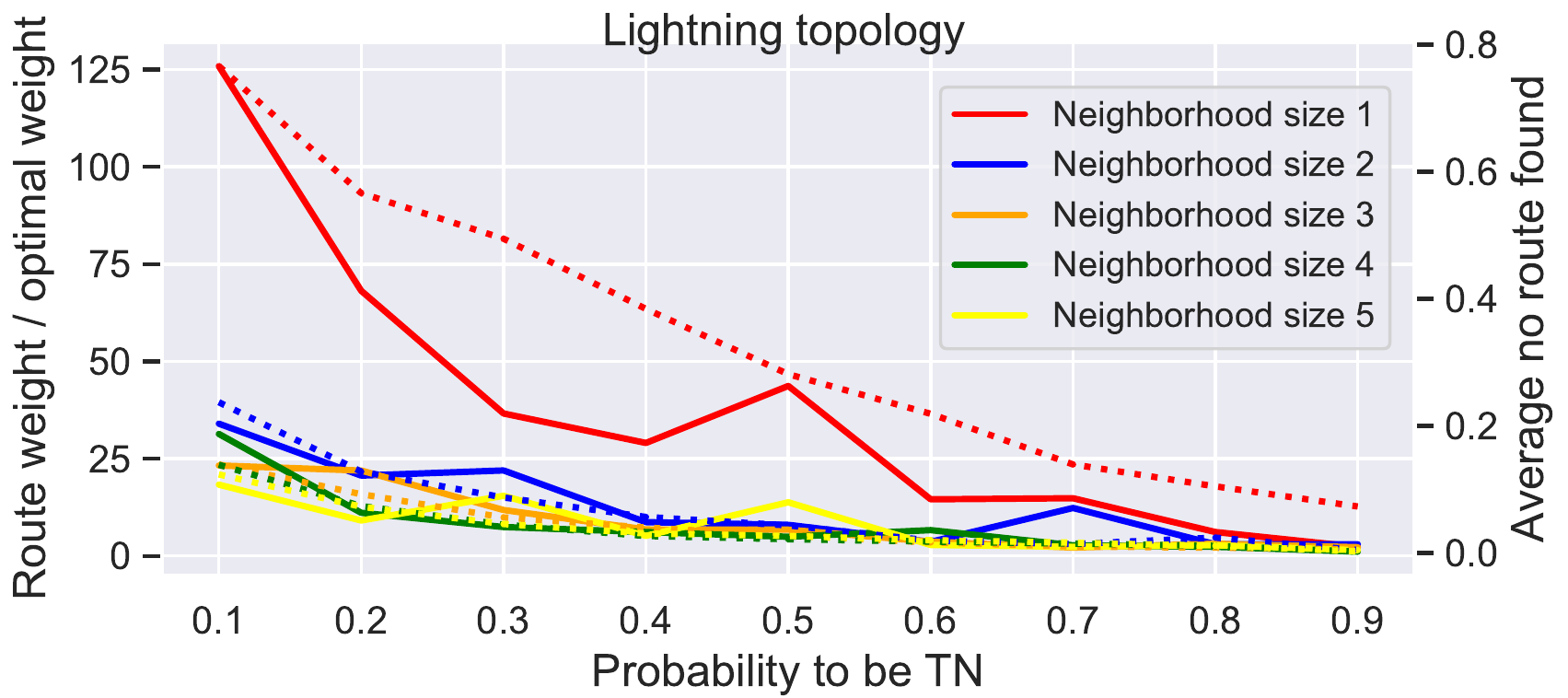}
	\caption{The probability to not find a route compared to the found route.}
	\label{figure::lightning_uniform_2d}
\end{figure}

\boldheader{Effectiveness-Confidentiality Tradeoff}
We next examine the RDA's effectiveness, i.e.,
the case in which the routes may be unavailable. 
As we showed in Lemma~\ref{lemma:effectiveness}, 
highly ineffective queries may result
when TNs offer unavailable routes, in which case 
the wallet will have to query many TNs and lose confidentiality. 
The next experiment explores this tradeoff by parametrizing, 
as before, the neighborhood size from which the wallet chooses the TN, 
and by the percentage of TNs in the network. We expect to see that as 
there are more TNs, the number of routes not found decreases (we will be offered more routes) and the number of queries will increase. 
For each pair of nodes in the topology, we ordered randomly all the TNs in the relevant neighborhood, 
and queried them one by one, each for $5$ routes. Only if all the $5$ routes failed, then we queried another one.

As we discussed before, the transactions in the Lightning network are private,
both in terms of the participating parties and the transaction size. 
In our experiment we hence simulate the unavailability of the nodes by considering only the lack of liquidity. For each channel, we simulated the occupied liquidity using a random variable $v \sim Uniform[0,1]$, and defined the already-locked liquidity to be $v \cdot tx_{size} \cdot factor$, while $tx_{size}$ is the transaction size that we try to route, which is $10^6$ millisatoshis. 
The x-axis in the graphs indicate this factor. 

In Figure~\ref{figure::sparse_priv_effectiveness}, we fixed the probability 
at which a channel will accept the transaction to $60\%$, and examined 
the number of queries that the wallet will have to execute before finding a route that accepts the transaction. We see that when we query bigger neighborhoods, the number of queries increases due to longer routes (which decreases the availability), but also more routes are found because there are more TN 
to query (and thus more \textit{disjoint routes}, i.e. different channels).

Figure \ref{figure::lightning_priv_effectiveness} shows this tradeoff 
on the Lightning network. As in the sparse network, we see that larger neighborhoods will result in more available routes, but at the cost of more queries. Note the curvature of the graph, compared to the sparse topology, which shows that the number of queries does not increase as fast when there are more TNs. We consider this as a result of the higher degree in the lightning network; the 3-neigborhood 
contains almost the entire graph, thus the TNs will yield disjoint routes which increases the probability to availability.

\begin{figure}
	\centering
	\includegraphics[scale=0.5]{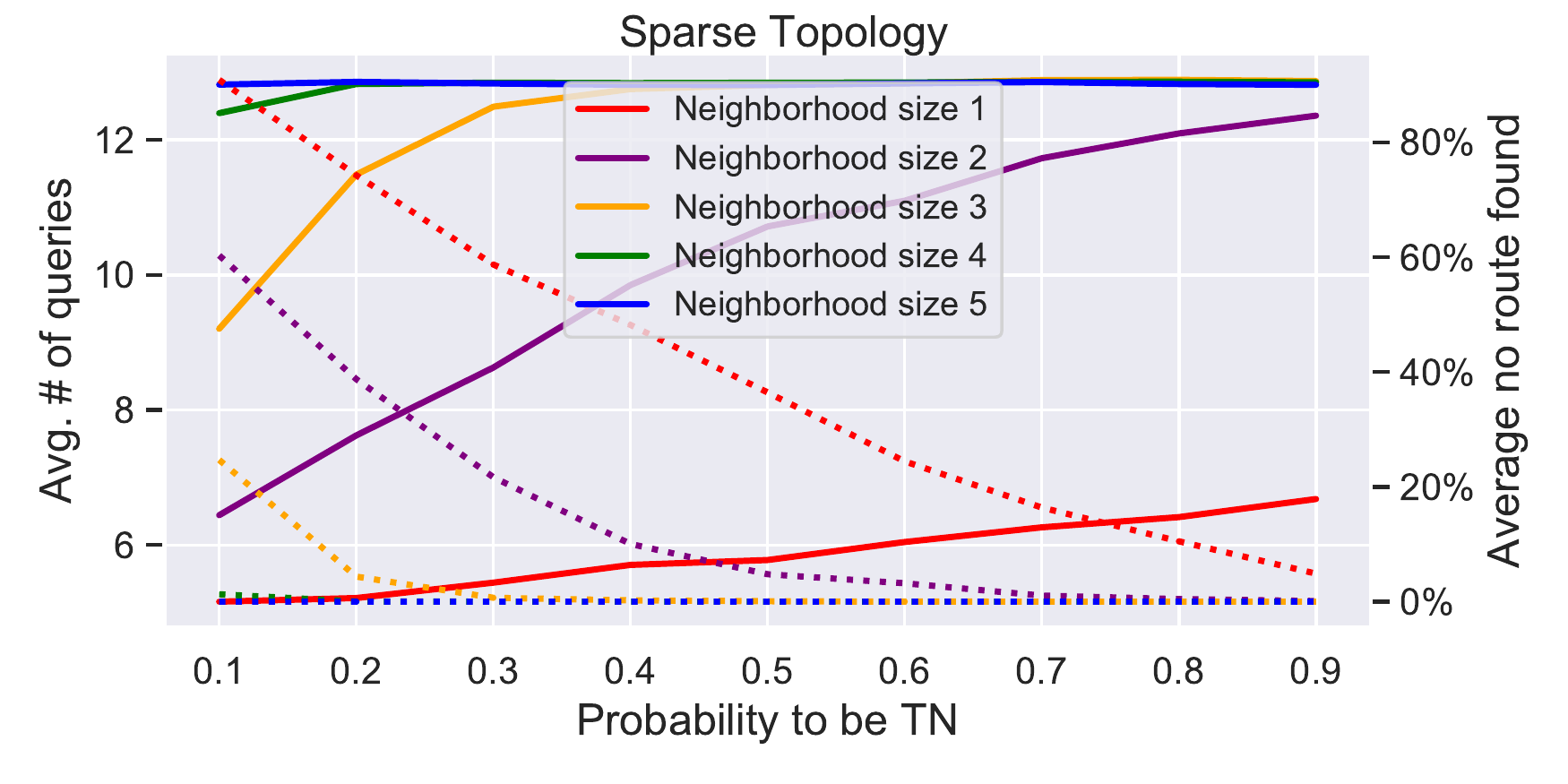}
	\caption{Privacy-effectiveness tradeoff.}
	\label{figure::sparse_priv_effectiveness}
\end{figure}
\begin{figure}
	\centering
	\includegraphics[scale=0.5]{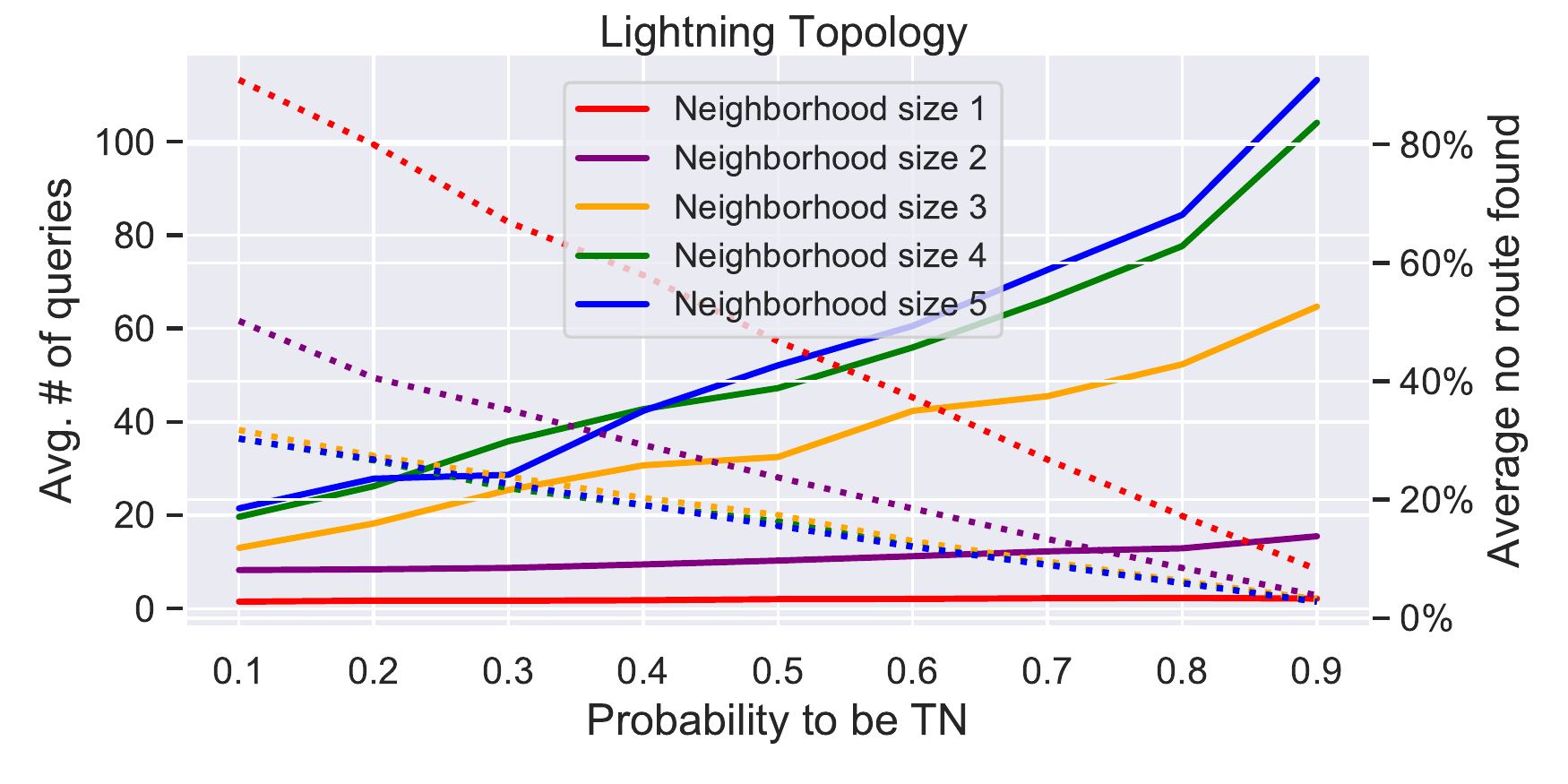}
	\caption{Privacy-effectiveness tradeoff.}
	\label{figure::lightning_priv_effectiveness}
\end{figure}

\boldheader{Effectiveness-Efficiency Tradeoff}
The last tradeoff discusses the average route weight when assuming that the routes could be unavailable. In this experiment, we tried to route through the 10 shortest paths, and checked the average weight and the number of pairs that did not find any route.

Figure~\ref{figure::sparse_fee_effectiveness} shows the tradeoff between the fee and the effectiveness of the route. Unlike the previous section, where we stopped querying after we found a valid route, in here we continued to ask all the TNs, and deduce the percentage of routes that we cannot use, and the average fee of the one that we can. 
This figure also reveals an interesting phenomenon in which the fee decreases when the probability for availability decreases. This happens because the available routes become shorter, and thus the average fee decreases.

Figure~\ref{figure::lightning_fee_effectiveness} shows this tradeoff in the Lightning network. We follow the same methodology of the previous section, in which each channel accepts a transaction randomly, as a function of the transaction size, a uniform generated number, and a factor. It is interesting to note that the last phenomenon that we described on the sparse topology does not hold here. The fee increases due to the larger number of pairs that succeed to transact. 
This might suggest that either 
there is a single route and if it is not available, then 
the transaction cannot be executed; or  there are many different 
routes with approximately the same weight, which makes the unavailability of some routes 
have a smaller effect on the route efficiency.

\begin{figure}
	\centering
	\includegraphics[scale=0.5]{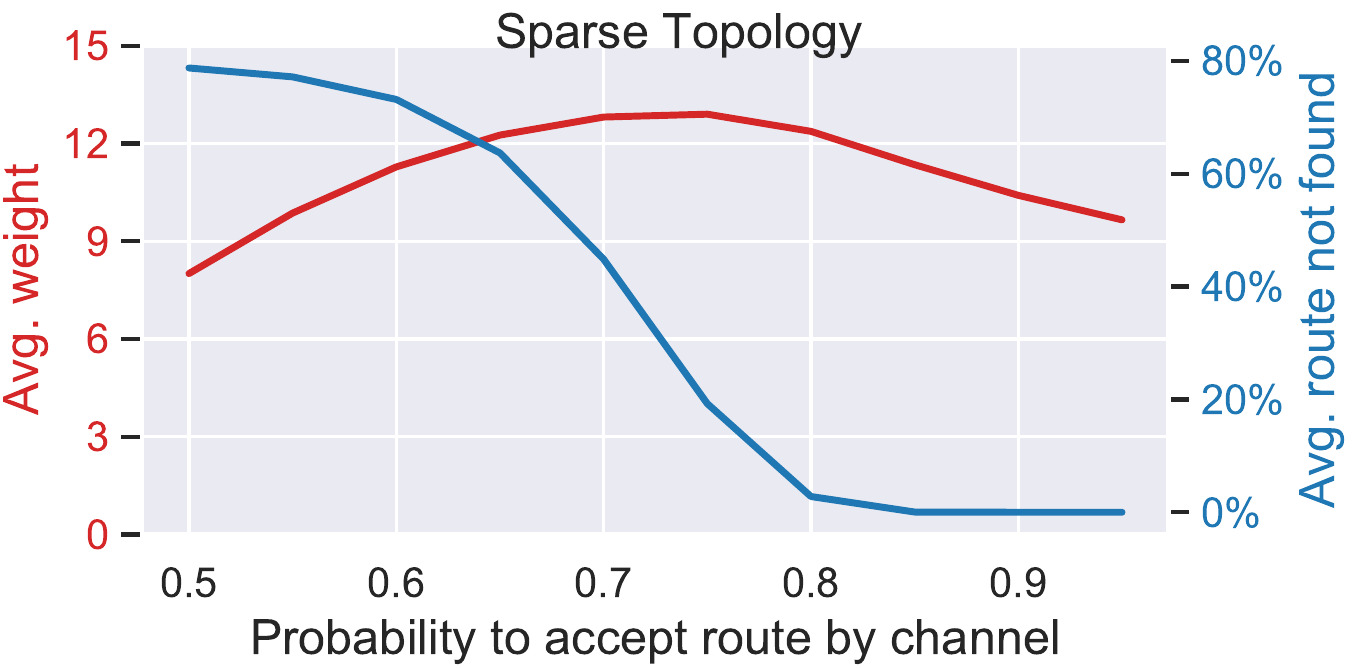}
	\caption{Efficiency-effectiveness tradeoff in the sparse topology.}
	\label{figure::sparse_fee_effectiveness}
\end{figure}
\begin{figure}
	\centering
	\includegraphics[scale=0.5]{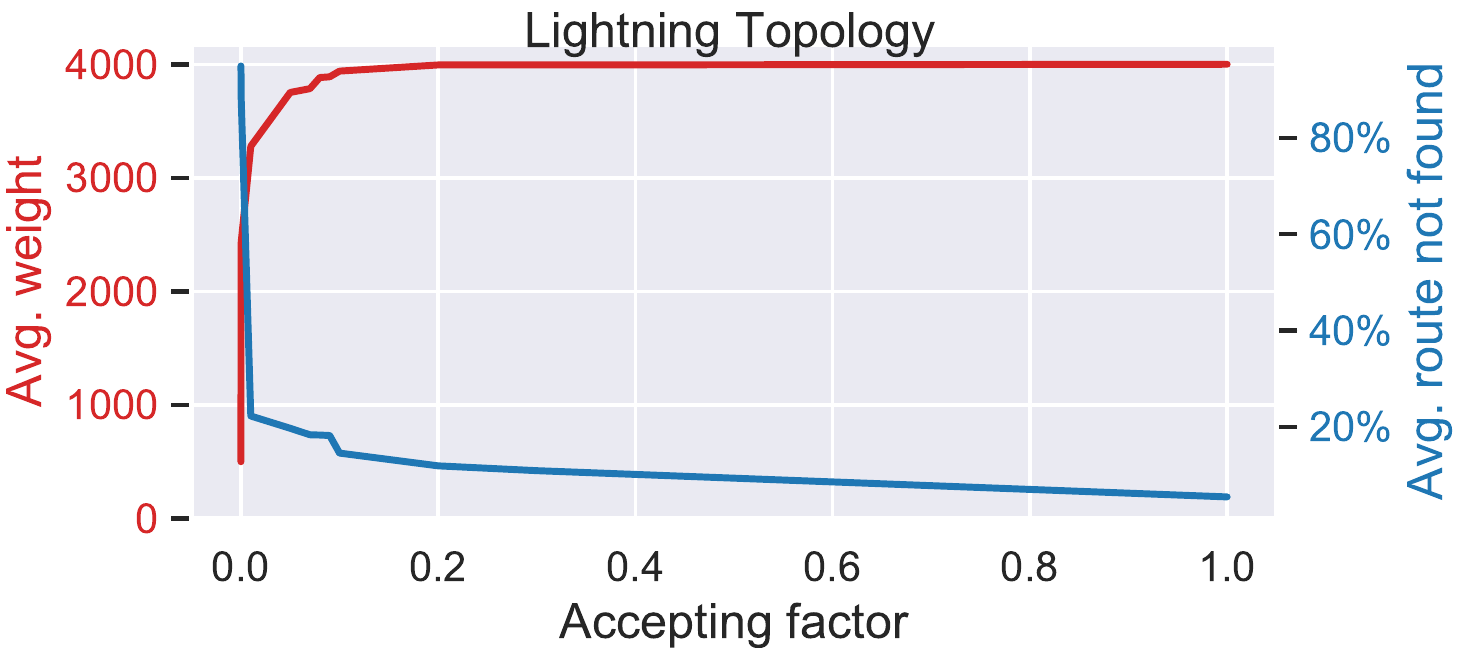}
	\caption{fee-effectiveness tradeoff in the Lightning network topology.}
	\label{figure::lightning_fee_effectiveness}
\end{figure}

\subsection{Extensions}

We next explore a generalization of
trampoline nodes, which may further
improve scalability. 
We observe that in principle,
nodes can answer path request queries 
also if they just used the path. 
The intuition is that while nodes will not remember 
the entire topology 
they may still have a 
route cache.
Let us hence consider ``Partial Nodes'' (PN),
nodes that share their past knowledge using the same selfish mechanisms 
like TNs, when answering path request queries.

Note that in this extension, the wallet nodes will not 
benefit from a better effectiveness (just like the TNs, the PNs do
not know the availability of the channels) or confidentiality 
(because they still query the same number of intermediate nodes). 
The major improvement that partial nodes will contribute to the network is the efficiency of the resulting
routes. The are more nodes in the network that share the optimal routes to the target, 
therefore the ``detour" of the route through the TN/PN will be smaller. Figure~\ref{figure::lightning_partial_result_ratio} assumes that each PN stores paths to 50 uniformly selected nodes. 
Here we fixed the number of TNs and the neighborhood size to obtain 
a clearer view of the improvement in the resulting routes.
On the other hand, Figure~\ref{figure::sparse_partial_result_ratio} 
shows the limited benefits of ``partial nodes" on the sparse topology. 
The effectiveness is low due to the small number of different routes to the target.

\begin{figure}
	\centering
	\includegraphics[scale=0.5]{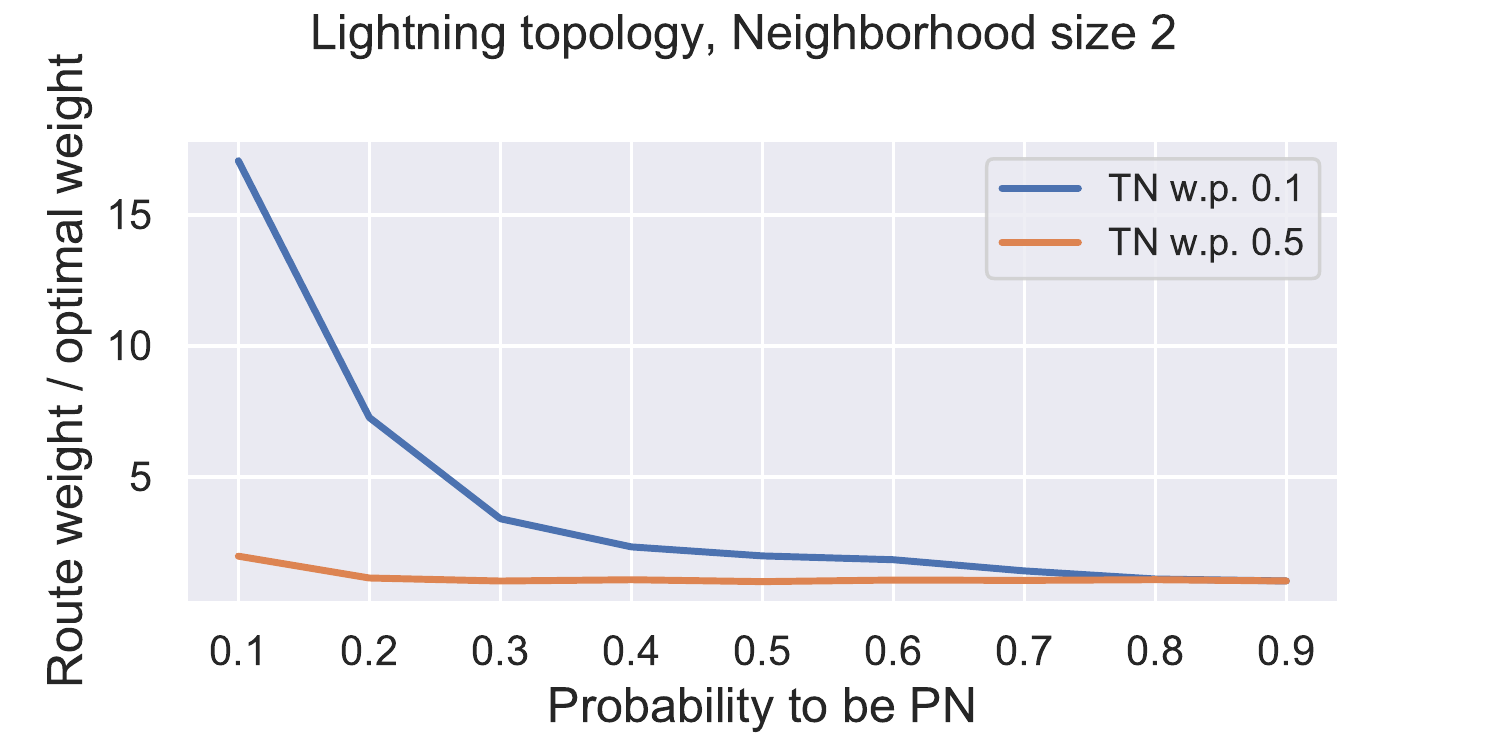}
	\caption{Ratio of the resulting route and the optimal 
	with Partial Nodes (the Lightning network topology).}
	\label{figure::lightning_partial_result_ratio}
\end{figure}
\begin{figure}
	\centering
	\includegraphics[scale=0.5]{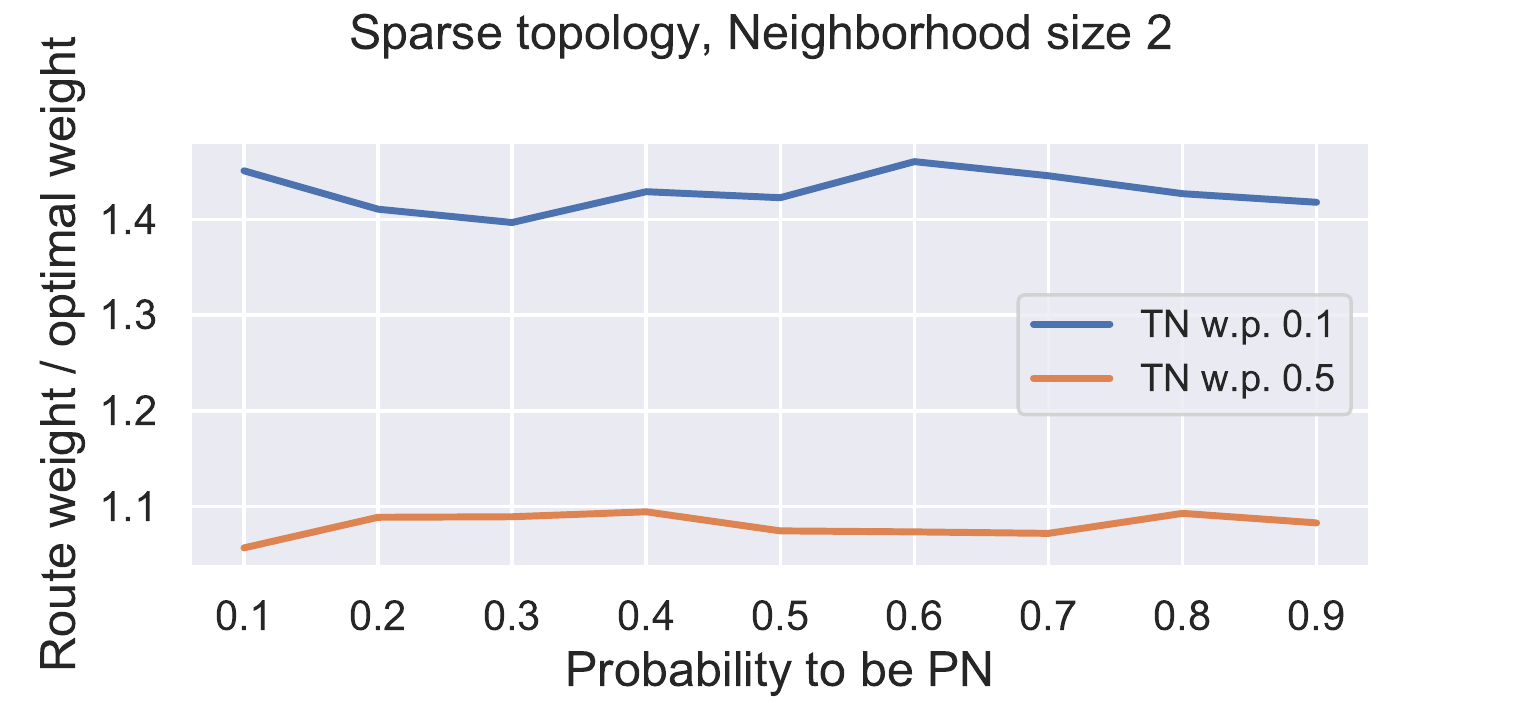}
	\caption{Ratio between the resulting route and the optimal 
	route with Partial Nodes (the sparse topology).}
	\label{figure::sparse_partial_result_ratio}
\end{figure}

Finally, it is also interesting to consider
the impact of altruistic nodes, i.e. nodes that hold
the entire topology and answer route queries unselfishly. In this case, all the nodes in the network can simply query the altruistic nodes, and get the optimal result without any cost beside the possible lack of privacy to a small number of entities.

\section{Related Work}  \label{sec::related_work}

Aspects of route discovery algorithms in off-chain networks were described in a SoK by Gudgeon et al. \cite{gudgeon2019sok},
identifying effectiveness, efficiency, scalability, cost-effectiveness 
and privacy as key challenges.
Following this, many suggestions were made in order to find a route discovery algorithm. Few among these suggestions are SpeedyMurmurs \cite{roos2017settling} and SilentWhispers \cite{malavolta2017silentwhispers} which focus on the effectiveness and efficiency (SpeedyMurmurs's simulations show that they find a route by up to two orders of magnitude faster then the existing algorithms, while maintaining the same success ratio). Unfortunately, in real-world networks, they face the challenge of cost-effectiveness that is driven by the selfishness of the nodes.

Another suggestion was made in SpiderNetwork \cite{sivaraman2018routing}, where the payments are split into units and the route discovery algorithm routes each of them individually (similar to packet switched network). This method however does not account for privacy and cost-effectiveness in selfish models. In MAPPCN \cite{tripathymappcn} the authors suggest and analyze a route discovery algorithm 
which preserves privacy but which does not account of the leaking
issue addressed in this paper. 

Further related research 
exists in other networks than off-chain networks,
e.g., on selfish routing in wireless networks, however,
there selfish nodes typically wish to avoid participating in forwarding packets in order to not waste energy. Another example is gaining source-target privacy in multi-hop networks such as TOR \cite{dingledine2004tor}.

From an economical perspective, much research exists on
markets with producers that offers digital goods (i.e. products which are infinitely expansible \cite{quah2003digital}) and consumers with search friction (i.e. cost in addition to the price \cite{anderson1999pricing}). One example is Pandora's problem \cite{weitzman1979optimal} wherein Pandora needs to pay a
fee to open boxes (which is equivalent in our model to loss of privacy or additional networking costs), and each box offers different value.

In general, our analysis is in the spirit of classic models such as
\cite{anderson1999pricing}, as we can model the
interaction between the TN (who sells an item or service) 
and the wallets that consume it. 
Many different scenarios of this game were researched previously (\cite{au2018competition}, \cite{choi2018consumer}), and the general result is that the pricing in the network will rise with the search friction.

\section{Conclusion \& Future Work} \label{sec::conculsions}

We modeled, analyzed and empirically evaluated the tradeoff
between efficiency, effectiveness and privacy of
route discovery in offchain networks which come
with scalability requirements and where node behave selfish.
In particular, we have shown that current solutions
can be inefficient in general, which raises interesting
avenues for future research.
Another interesting direction for future research
regards the exploration of further strategic behaviors.

We see a trampoline node as an interesting player with different economical 
incentives compared to the other nodes in the network.
While the wallet has a search friction that is 
based on the number of queries that it needs to perform and the 
resulting privacy loss, 
the TNs may offer different prices that may be changed according to the wallet's strategy. 
The study of the resulting strategic behaviors and games may provide interesting
insights, e.g., on whether overestimating the value of privacy 
will motivate the TN to increase the fee of the offered route.
Moreover, the interaction between the TNs (the ``sellers") and the wallets (the ``buyers") 
falls into the area of trading ``digital goods", i.e. unbounded amount of products to sell, 
which is well researched. As far as we know, this research was never combined with search-friction 
for the buyers, as in our case. This indeed makes sense in ``traditional" digital goods 
(such as servers that offer files to download), but when considering privacy and bandwidth, 
as in this paper, this could be a very interesting extension.
It could further be interesting to 
broaden the discussion to strategic behaviors which lead
TNs to route transactions through worse routes in order to manage liquidity in the TN's channels.

\bibliographystyle{plain}
\bibliography{references}

\end{document}